\documentclass[english]{article}
\usepackage[T1]{fontenc}
\usepackage[latin9]{inputenc}
\usepackage{amsthm}
\usepackage{amsmath}
\usepackage{graphicx}

\makeatletter
\numberwithin{equation}{section}
\numberwithin{figure}{section}
\newcommand{\lyxaddress}[1]{
\par {\raggedright #1
\vspace{1.4em}
\noindent\par}
}
\theoremstyle{plain}
\newtheorem{thm}{\protect\theoremname}
  \theoremstyle{definition}
  \newtheorem*{example*}{\protect\examplename}
  \theoremstyle{remark}
  \newtheorem*{rem*}{\protect\remarkname}
  \theoremstyle{remark}
  \newtheorem{rem}[thm]{\protect\remarkname}

\makeatother

\usepackage{babel}
  \providecommand{\examplename}{Example}
  \providecommand{\remarkname}{Remark}
\providecommand{\theoremname}{Theorem}

\begin{document}

\title{An Optimal Linear Coding for Index Coding Problem }

\author{Pouya Pezeshkpour }

\date{Department of Electrical Engineering, Sharif University of Technology}

\maketitle

\lyxaddress{\begin{center}
Email: pezeshkpour\_pouya@ee.sharif.edu
\par\end{center}}
\begin{abstract}
An optimal linear coding solution for index coding problem is established.
Instead of network coding approach by focus on graph theoric and algebraic
methods a linear coding program for solving both unicast and groupcast
index coding problem is presented. The coding is proved to be the
optimal solution from the linear perspective and can be easily utilize
for any number of messages. The importance of this work is lying mostly
on the usage of the presented coding in the groupcast index coding
problem.
\end{abstract}

\section{INTRODUCTION}

We consider the problem of Index Coding when the sender wish to communicate
$n$ messages $w_{i}$, $i\epsilon\{1,...,n\}$ to $m$ receivers
$r_{j}$, $j\epsilon\{1,...,m\}$with desire to recover a subset of
messages $R_{j}$, with prior knowledge of $S_{j}$, as a side information
which are subset of messages.

If we consider the encoder as a function $f(w_{1},...,w_{n})$ which
role is to satisfy the desire of ever single one of the receivers,
the index coding problem's goal is to minimize the size of $f$'s
range.

To do so, we characterize index coding problem as a bipartite graph
where the first set of nodes represent indices of the messages and
the other one represent indices of the receivers. The edges in this
graph represent the knowledge of the receiver node from the message
node as a side information. By using the same definition as {[}4{]},
we represent the desire of each receiver as $r_{j}(R_{j}\mid S_{j})$.
As an example, consider an index coding problem with 3-message index
and 2 receiver with $S_{1}=\{2,3\}$, $S_{2}=\{1\}$, $R_{1}=\{1\}$
and $R_{2}=\{2,3\}$. As a result this problem is represented as:

\begin{equation}
r_{1}(1\mid2,3),r_{2}(2,3\mid1),\label{eq:1}
\end{equation}

Note that according to graphical characterization presented earlier,
this example can be illustrated as Fig. 1.

The remainder of this paper is organized as follows. In Section 2,
the coding solution for unicast cases and an algorithm to simplify
its performance is provided. The generalization of this coding program
for groupcast cases is provided in Section 3. Section 4 concludes
the paper.

\begin{figure}
\includegraphics[bb=-100bp 0bp 207bp 120bp]{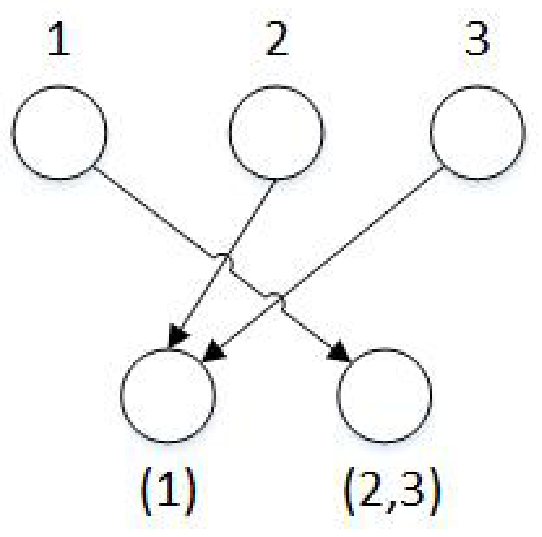}

Fig. 1. The graph representation of 3-message index coding problem.
\end{figure}

\section{THE CODING PROGRAM FOR UNICAST CASES}

Lets consider a general unicast index coding problem as following:

\begin{equation}
r_{1}(1\mid S_{1}),r_{2}(2\mid S_{2}),...,r_{n}(n\mid S_{n}),\label{eq:2}
\end{equation}

By considering the graph characterization for this problem, we call
two receiver nodes a cross neighbor if the desire message in the first
one be part of the other one's side information and vice versa. As
a result, a cross neighbor set is a set of receiver nodes which any
two of them are a cross neighbor. It is obvious that we can satisfy
the desire of every nodes in a cross neighbor set by sending the summation
of all the desired messages in the set. 

The key to the optimal coding solution is to find the minimum number
of cross neighbor sets in the graph. As a result, our coding program
is to partition the graph to the minimum number of cross neighbor
sets and then send the summation of all the desired nodes in the each
set to satisfy their desire and separately send every desired messages
which are not in any cross neighbor set.
\begin{thm}
The presented coding program is an optimal solution for the index
coding problem from linear perspective.\end{thm}
\begin{proof}
As it is mentioned before, we can satisfy the desire of every nodes
in a cross neighbor set by sending one message through the channel.
As a result, to get to the optimal program we must send this message
for every single one of the cross neighbor sets. Since, we considered
the minimum number of existing cross neighbor sets, it is clear that
this approach will satisfy the maximum number of receivers with the
minimum number of outputs. Now that we satisfy the cross neighbor
sets we can omit them from our original graph without effecting the
final result. The remaining graph has no cross neighbor. As a result,
by a single transmission we can at most satisfy one receiver node,
so it is optimal to separately send all the desired messages in the
remaining nodes through the channel. 
\end{proof}

\subsection{AN ALGORITHM FOR UTILIZING PRESENTED CODING PROGRAM}

The only Achilles heel in this program, is the fact that partitioning
the graph to the minimum number of cross neighbor sets is not a simple
task. To do so, we first need to specify all the cross neighbors in
our graph. Then, by defining a new graph which its nodes are all the
cross neighbor nodes and its edges are between every two nodes which
were cross neighbor, we can transform our cross neighbor sets to complete
subgraphs. Therefore, instead of finding the minimum number of cross
neighbor sets we must find the minimum number of complete subgraphs
which previously considered thoroughly in graph theory.
\begin{example*}
Lets consider the following example:

\begin{equation}
r_{1}(1\mid2,3,4),r_{2}(2\mid5),r_{3}(3\mid1,4),r_{4}(4\mid1,3),r_{5}(5\mid2,6),r_{6}(6\mid4),\label{eq:}
\end{equation}

This instance can be illustrate as Fig. 2. To perform the presented
coding program we must partitioning the graph to the minimum number
of cross neighbor sets, and since this case is not very complicated
we can obtain this goal without using provided algorithm. The partitioning
of the graph to the cross neighbor sets is provided in Fig. 3. As
a result, we can solve this index coding problem by sending $w_{1}+w_{3}+w_{4},w_{2}+w_{5},w_{6}$
through the channel. 

\begin{figure}
\includegraphics[bb=0bp 0bp 442bp 150bp]{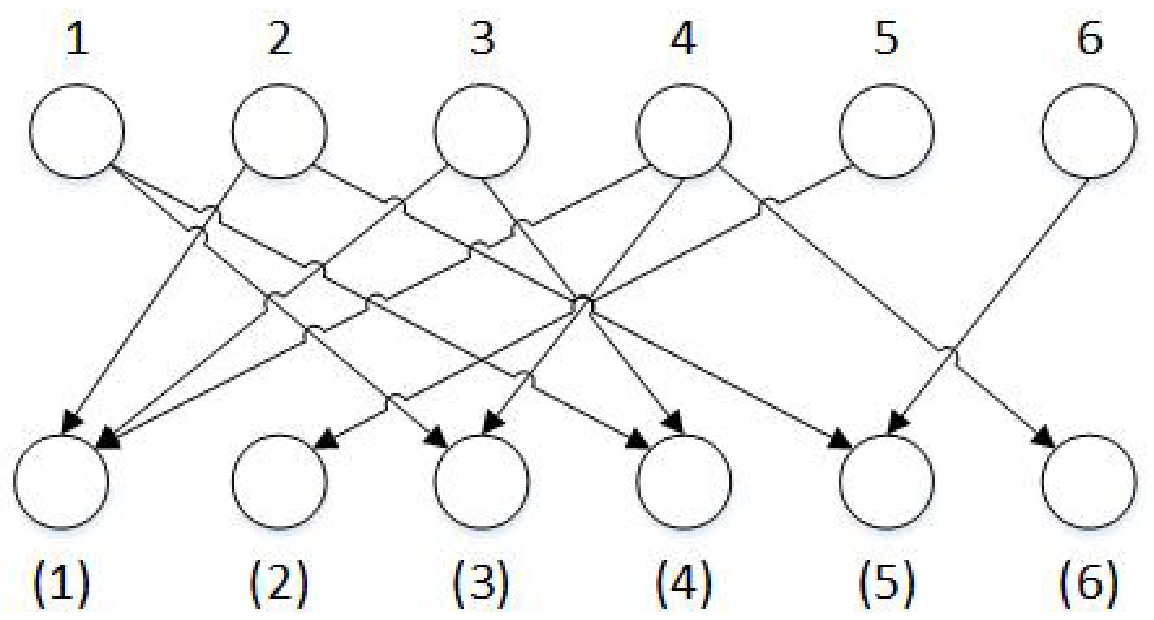}

Fig. 2. The graphical illustration of the example.

\end{figure}

\begin{figure}
\includegraphics[bb=0bp 0bp 442bp 150bp]{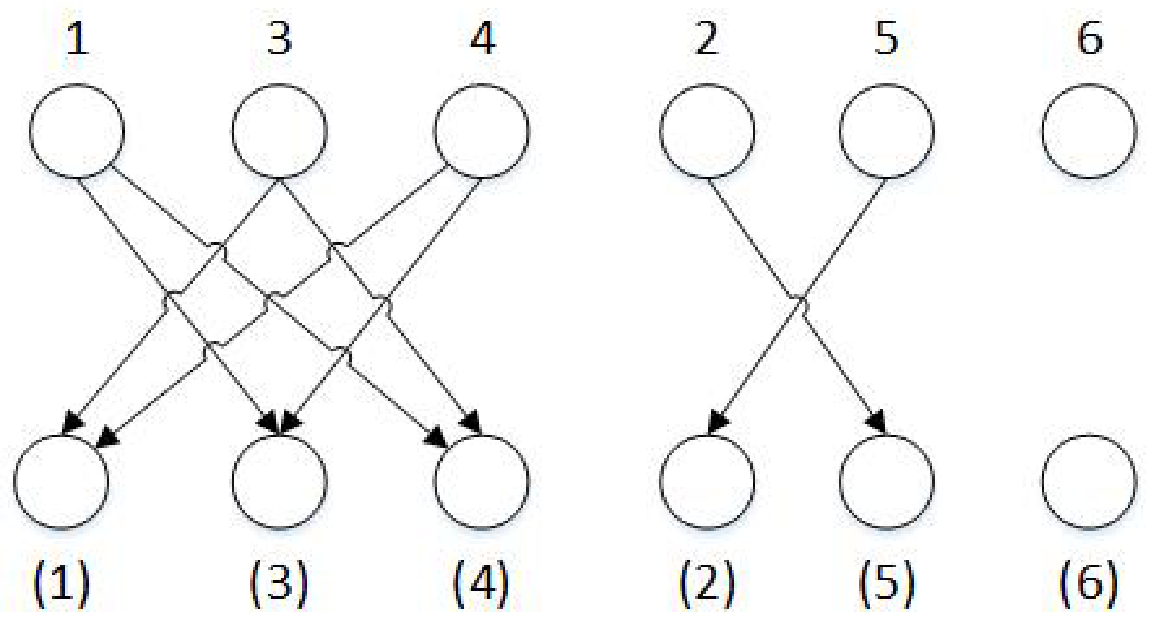}

Fig. 3. The partitioning of the graphical characterization of the
example to the cross neighbor sets.
\end{figure}

\end{example*}

\section{THE CODING PROGRAM FOR GROUPCAST CASES}

Lets consider a general groupcast index coding problem as following:

\begin{equation}
r_{1}(R_{1}\mid S_{1}),r_{2}(R_{2}\mid S_{2}),...,r_{n}(R_{n}\mid S_{n}),\label{eq:2-1}
\end{equation}

To perform the coding program presented in the last section, we first
need to characterize the problem to the same format as considered
in the previous section. To do so, we need to treat every single desired
messages as a separate receiver node, it means that we must break
down the $r_{i}$ to $r_{i1},...,r_{i\mid R_{i}\mid}$which have the
same side information as $r_{i}$. As a result, the total number of
receiver nodes would be equal to $\sum_{i=1}^{i=n}\mid R_{i}\mid$.
The reason behind this characterization, is the fact that to get to
the optimal solution we need to consider all the desired messages
in every receivers in the same time.

Now that we transform the problem to a unicast case, By doing the
same method as the previous section, and probably use the presented
algorithm to finalize the coding, since the number of nodes will grow
much larger, we can easily get to the optimal solution.
\begin{rem*}
It is obvious that we cannot satisfy two or more desire of a receiver
node with a single transmission. As a result, breaking down the receiver
nodes will not affect the optimal solution.\end{rem*}
\begin{rem}
Note that although in this case we will have different nodes with
the same desire, since we build our coding program on finding the
minimum number of cross neighbor sets, considering these same desires
as separate receiver nodes will not affect the fact that our coding
program is optimal.
\end{rem}

\section{CONCLUSION}

In this work by considering the linear concept of index coding problem,
an optimal coding program for unicast and groupcast cases provided.
The most important privilege of this coding is the fact that it can
be used in the general case of index coding problem. Although performing
this method may be difficult in cases with lots of messages, but presented
algorithm will simplify the usage of this coding solution.

\end{document}